\documentclass[conference]{IEEEtran}
\IEEEoverridecommandlockouts

\usepackage{amsmath}
\usepackage{bm}
\usepackage{mdframed}
\usepackage{amsthm}
\usepackage{mathtools}
\usepackage{comment}

\usepackage[utf8]{inputenc}
\usepackage{graphicx}
\usepackage[colorinlistoftodos]{todonotes}
\usepackage{amssymb}
\usepackage{longtable}
\usepackage{nohyperref}
\hypersetup{bookmarks=false}
\usepackage[noadjust]{cite}
\allowdisplaybreaks
\usepackage{lineno}
\usepackage{lipsum}
\usepackage{multicol}
\usepackage{subcaption}
\usepackage{enumitem}

\newtheorem{corollary}{Corollary}
\newtheorem{assumption}{Assumption}
\usepackage{xcolor}

\usepackage[linesnumbered,ruled,vlined]{algorithm2e}

\def\grad{\nabla}

\SetKwInput{KwInput}{Input}                
\SetKwInput{KwOutput}{Output}              
\SetKwInput{KwInitialize}{Initialize}                

\begin{document}
\bstctlcite{IEEEexample:BSTcontrol}
 \title{Reservation of Virtualized Resources with Optimistic Online Learning}

\vspace{-4mm}
\author{\IEEEauthorblockN{Jean-Baptiste Monteil\IEEEauthorrefmark{1}, George Iosifidis\IEEEauthorrefmark{2}, Ivana Dusparic\IEEEauthorrefmark{1}}
		\IEEEauthorrefmark{1}School of Computer Science and Statistics, Trinity College Dublin\\
		\IEEEauthorrefmark{2}Delft University of Technology, Netherlands}

\maketitle

\vspace{-2mm}
\begin{abstract}
The virtualization of wireless networks enables new services to access network resources made available by the Network Operator (NO) through a Network Slicing market. The different service providers (SPs) have the opportunity to lease the network resources from the NO to constitute slices that address the demand of their specific network service. The goal of any SP is to maximize its service utility and minimize costs from leasing resources while facing uncertainties of the prices of the resources and the users' demand. In this paper, we propose a solution that allows the SP to decide its online reservation policy, which aims to maximize its service utility and minimize its cost of reservation simultaneously.  
We design the Optimistic Online Learning for Reservation (OOLR) solution, a decision algorithm built upon the Follow-the-Regularized Leader (FTRL), that incorporates key predictions to assist the decision-making process. Our solution achieves a $\mathcal{O}(\sqrt{T})$ regret bound where $T$ represents the horizon. We integrate a prediction model into the OOLR solution and we demonstrate through numerical results the efficacy of the combined models' solution against the FTRL baseline.

\end{abstract}
\IEEEpeerreviewmaketitle


\section{Introduction}


\textbf{Motivation}. The virtualization of wireless networks has gained significant interest in recent studies, cf. \cite{5g-ppp, mano-nfv}. This new technology enables the development of the Network Slicing framework, where service providers (SPs) can lease virtualized network resources from the Network Operator (NO) to address the demand of their specific network service \cite{foukas-commag}. 
Network Slicing promises to boost the utilization efficiency of the network resources by accommodating multiple and diverse SPs on the NO's infrastructure. This in turn brings new challenges: on the one hand the NO must accommodate heterogeneous slices on its network to satisfy diverse requirements of the SPs; on the other hand the SPs must request network resources or slice requirements in a smart and proactive way by anticipating their future demand.

The players are expected to operate in a real-time market, where the SPs can lease both computing and storage resources while the NO offers both in-advance reservation and on-the-fly spot opportunities. The modeling of such slicing market draws ideas from cloud marketplaces \cite{amazon-reserved, amazon-spot}, where the Cloud Provider allows customers to bid for resources in the on-demand and spot markets \cite{carlee-infocom18, carlee-sigcomm15}. This market will allow the NO to proactively schedule the slice configuration based on the information coming from the in-advance reservation requests, but also offer the available spot resources dynamically, leading to slice re-configuration and boosting network utilization.

In this context, one SP competes with other SPs for the network resources in the on-demand and spot markets and must request/bid for the resources while ignorant of their prices. 
We expect the NO to reveal those prices after the SP request. 
Therefore, the SP must decide its requests dynamically without the information of the resource pricing and its own future demand. Additionally, we expect the prices to vary according to non-stationary patterns, as they might depend on multiple underlying factors, such as the other SPs requests, the NO internal needs, etc. We highlight here the necessity for the SP to build a decision model robust to uncertainty, while being able to use its own historical demand and the NO's feedback about historical prices.

\textbf{Related Work}. 
By anticipating the resource utilization, the NOs can enhance their resource management decisions regarding resource provisioning or allocation.
In \cite{DeepCog}, network traffic information is leveraged to plan the capacity needed for each slice in a multi-tenant framework. Using a data-driven approach including C-RAN, MEC and core networks, the solution outperforms other state-of-the-art deep learning solutions \cite{infocom17}, \cite{mobihoc18}. The approach in \cite{oliveiraTNSM} employed an adaptive forecasting model of the elastic demand for network resources to perform slice allocation in Internet Access Services. The authors in \cite{XaviINFOCOM17} and \cite{XaviTrans19} predicted the required resources by tenants for the future time window to perform slice requests admission and schedule the users' traffic within each slice. 
In \cite{cui-iccc20} cellular traffic prediction helps the allocation policy for the vehicular network slice. \cite{jb-icc20} uses historical traffic to design the SP resource reservation policy. 
Unlike these approaches, our solution does not need offline training and provides performance guarantees against all types of traces.

Recent works consider the SP resource provisioning problem. The paper \cite{reyhanian} developed a two-time scale approach for the activation and the re-configuration of the slices while considering the reservation of both RAN and backhaul resources. \cite{zhang-tcom2018} focuses on wireless spectrum considering two reservation schemes (in advance and on demand). In \cite{vincent-TVT2020}, the authors develop a two-stage approach for the resource reservation and the intra-slice resource allocation. These works presume a stationary environment where user statistics do not change and/or cost of resources are supposed constant. This paper differs from our previous works \cite{JB2021, JB2022}, as we now use prediction to support the reservation model.



\textbf{Methodology and Contributions}.
The problem of learning how to bid in an online manner while facing uncertainty fits to the Online Convex Optimization (OCO) framework, introduced by Zinkevich \cite{zinkevich}. In OCO, the learner tries to minimize its total loss with respect to the best static solution:
\begin{align}
    R(T) = \sum_{t=1}^T f_t(\bm z_t) -  \min_{\bm z \in \mathcal{Z}} \sum_{t=1}^T f_t(\bm z), \label{eq:static_regret}
\end{align}
by deciding the reservation vector $\bm z_t$ at each round $t$, without knowing the convex loss $f_t$. We say the online policy $\{\bm z_t\}_{t=1}^T$ has \emph{no-regret} if the achieved regret is sublinear, i.e. $R(T) = o(T)$, in other words $\lim_{T\rightarrow \infty} R(T)/T = 0$. We build our Optimistic Online Learning for Reservation (OOLR) solution upon the Follow-The-Regularized Leader (FTRL) algorithm \cite{ftrl}. 
We develop an \emph{optimistic} version of the FTRL, first introduced by Rakhlin and Sridharan \cite{sridharan}, where the decision relies on an adaptive proximal regularizer term and the optimistic term of the next gradient prediction $\grad \hat f_{t+1}(\hat{\boldsymbol{z}}_{t+1})$. With perfect predictions, the regret of our decisions reduces to $\mathcal{O}(1)$, synonymous of negative regret. With arbitrarily bad predictions (of the order of $T$), the regret bound is $\mathcal{O}(\sqrt{T})$.

As the SP accumulates historical data about prices and demand, there exists the possibility to extract predictions for the next slot values based on previous window of the traces by using auto-regressive methods. Holt-Winters, Auto-Regressive Integrated Moving Average (ARIMA) or Neural Networks have been applied in \cite{oliveiraTNSM} and in \cite{XaviINFOCOM17}. Albeit accurate, these methods do not provide performance guarantees. 
The ARMA-OGD algorithm presented in \cite{anava} is an accurate, robust and computationally low prediction model. It generates the predictions through an auto-regressive process, where the lag coefficients are updated using the online gradient descent (OGD) method, which has low time complexity. It also provides regret guarantees against the best Auto-Regressive Moving Average (ARMA) predictor with full hindsight of the future. 

The contribution can be stated as follows:
\begin{itemize}[leftmargin=4mm]
\item we formulate an optimization problem for the SP where it aims to maximize the leased slice utility and minimize the reservation cost in the long-term; 
\item to solve the reservation problem faced by the SP, we develop an online learning solution (OOLR) which incorporates the \emph{optimistic} prediction of the next slot gradient;
\item we provide regret bound guarantees of $\mathcal{O}(\sqrt{T})$ for arbitrarily bad predictions and $\mathcal{O}(1)$ for perfect predictions;
\item we implement a prediction module to assist our OOLR decision algorithm and we demonstrate good performance of the combined solution, named OOLRgrad; 
\item against real world data and non-stationary traces, our OOLRgrad solution outperforms the FTRL baseline. We extend our model to the situation the NO only fulfills part of the SP reservation request due to capacity constraints. 

\end{itemize}

\vspace{2mm}
\section{Model and Problem Statement} \label{sec:model}




\textbf{Network and Market Model}. The key parameters of our model and solution are summarized in table \ref{table:notation} below.
We consider a slotted system $\{1,\ldots,T\}$.
A Network Operator (NO) sells virtualized resources to the service provider (SP), and we denote with $\mathcal{H}$ the set of $m=|\mathcal{H}|$ types of resources that comprise each slice. For instance, $\mathcal{H}$ may include bandwidth capacity, backhaul link capacity, edge computing and storage resources ($m=4$). The SP can reserve multiple kinds of resources which orchestration will enable the operation of the slice. We denote the in-advance reservation and spot reservation decisions at slot $t$ respectively as $\bm x_t = [x_1, \ldots,x_m]_t^\top$ and $\bm y_t = [y_1, \ldots,y_m]_t^\top$.
The optimal mix of resources composing the slice is unknown to the SP, as it depends on the type of request the SP receives from its users. Moreover, the benefit from each resource can be time-varying, e.g. bandwidth capacity can change due to varying channel conditions. 
The benefit from reservation $\bm x_t$ ($\bm y_t$) is quantified by the scalar $\bm x_t^\top \bm \theta_t$ ($\bm y_t^\top \bm \theta_t$), where the items of $\bm \theta_t \in \mathbb R^m$ are the individual contributions of each resource on the performance at slot $t$.



The utility stemming from such reservation scheme is non-linear. We model the slice utility of the SP as an increasing concave function of the acquired resources by using the logarithm function. For instance, the paper \cite{paschos-infocom18} provides the general form of $\alpha$-fair utility functions:
\begin{align}
    f(\bm z) = \left\{
    \begin{array}{ll}
        \frac{\bm z^{1-\alpha}}{1-\alpha} & \alpha \neq 1 \\
        \log(\bm z) & \alpha = 1
    \end{array}
\right.
\end{align}
where $\bm z$ is the reservation vector of slices. In \cite{srikant}, the utility from allocating bandwidth $x$ to a certain network flow $f$ is modeled as $a_f\log(x_f)$, where $a_f$ is a problem (and flow)-specific parameter. 
The logarithm function allows us to model as well the diminishing returns which naturally arise with the over-reservation of the network resources. For instance, the data rate is a logarithmic function of the spectrum; the additional revenue of the SP from more slice resources is typically diminishing. We model the slice utility function as a logarithmic concave function, weighted by the SP demand $a_t$, i.e. $a_t\log(1+\bm \theta_t^\top(\bm x_t + \bm y_t))$. 

The market operates in a hybrid model. At the beginning of each slot, the SP can lease network resources, plus additional resources on a spot market. We denote with $\bm p_t = [p_1,\ldots,p_m]^\top_t \in \mathbb R_+^m$ the unit price of the network resources; and we denote with $\bm q_t=[q_1,\ldots,q_m]^\top_t \in \mathbb R_+^m$ the unit price of the resources available in the spot market. The SP reservation policy consists of the reservation decision $\bm x_t$ and the spot decision $\bm y_t$. At the beginning of each slot $t$, the SP decides its $t$-slot reservation plan $(\bm x_t, \bm y_t)$, and pays the price $\bm p_t^\top \bm x_t + \bm q_t^\top \bm y_t$ at the end of the slot.

The NO can impose upper limits on the requests of the SP. For instance, the reservation request for resource $i$ must belong to the set $\Gamma_i = [0,D_i]$, where $D_i$ is the limit imposed by the NO on resource $i$. Therefore, the SP request will belong to $\Gamma_1 \times\ldots\times \Gamma_m$, which we denote $\Delta$. Such limitations arise from natural capacity constraints of the network, in charge of multiple services and its own needs. In some cases, the NO can be unable to fulfill the SP request, especially when the network is congested due to high users' demand load and heavy SPs requests. The NO must guarantee a certain Service Level Agreement (SLA), which we relate to the respect of a certain threshold ratio of the requested amount resource. For instance, the NO must deliver at least $\alpha=80\%$ of the desired capacity for the resource. We envision this scenario as an extension and assume from now that the NO must comply with the whole request if it belongs to $\Gamma_i$.




\textbf{Problem statement}.
Putting the above together, the ideal reservation slice policy is the solution of the following convex program:

\begin{align}
(\mathbb P):\quad \max_{ \{\bm x_t,\{\bm y_t\}\}_{t=1}^T } & \sum_{t=1}^T \Big(V  a_t \log((\bm x_t + \bm y_t)^\top \bm \theta_t + 1) \notag\\
&- (\bm p_t^\top \bm x_t +\bm q_t^\top \bm y_t) \Big) \label{prob-obj} \\
\text{s.t.}\quad   \bm y_t\in &\Delta, \quad \forall t=1,\ldots,T,  \label{prob-const2} \\
	 \bm x_t \in &\Delta, \quad \forall t=1, \ldots, T. \label{prob-const3}
\end{align}

In Objective (\ref{prob-obj}), we recognize the weighted sum of the slice performance (logarithmic term) and the payments (linear term). The latter term has a minus sign as the SP seeks to minimize its monetary cost. We sum over the number of slots $T$, as the goal is to maximize this weighted sum in the long-term.
Constraints (\ref{prob-const2}) and (\ref{prob-const3}) ensure the decisions belong to the constraint convex set $\Delta$. We define the hyper-parameter $V\geq1$ which balance the influence between the two terms (utility term and cost term). The bigger $V$, the more we favor the slice utility in the detriment of the cost of reservation.


$(\mathbb P)$ is a convex optimization problem but cannot be tackled directly due to the following challenges:
\begin{itemize}
    \item the users' demand $\{a_t\}$ is unknown, time-varying and non-stationary;
    \item the unit prices $\{\bm q_t\}$ and $\{\bm p_t\}$, are unknown, time-varying and non-stationary;
\end{itemize}

Due to these challenges, the convex problem $(\mathbb P)$ cannot be solved at $t=1$ for the next $T$ slots. Henceforth we define the loss function, at each slot $t$:
\begin{align}
    f_t(\bm x_t, \bm y_t) = - V a_t \log((\bm x_t + \bm y_t)^\top \bm \theta_t + 1)\\
    + (\bm p_t^\top \bm x_t +\bm q_t^\top \bm y_t)
\end{align}
The function $f_t$ is convex which allows us to use the OCO framework. Our goal is to decide at each slot $t$ the reservation plan $\bm z_t = (\bm x_t, \bm y_t)$ and achieve in the long term a sublinear static regret as defined in \eqref{eq:static_regret}.

\begin{table}
\caption{Key parameters and variables}
\scriptsize
	\centering%
	\begin{tabular}{|c|c|}
		\hline %
		\hline
		Symbol & Physical Meaning\\
            \hline %
		$m$ & Number of network resources composing a slice\\
            \hline
		$\bm x_t$ & Reservation in advance market in slot $t$ \\
            \hline
            $\bm y_t$ & Reservation in spot market in slot $t$ \\
		\hline %
            $\bm \theta_t$ & Contribution vector in slot $t$ \\
            \hline
		$a_t$ & User needs for the SP service in slot $t$\\
		\hline
		$\bm p_t$ & Unit price vector of the network resources at $t$ \\
		\hline
		$\bm q_t$ & Spot price vector for slot $t$\\
		\hline
		$T$ & Number of slots/horizon\\
		\hline
		  $D_i$ & Upper-bound imposed by the NO for reservation of resource $i$ \\
		\hline
		$\Gamma_i$ & $\Gamma=[0,D_i]$, feasible set for reservation of resource $i$\\
		\hline %
		$\Delta$ & Compact convex set $\Gamma_1\times\ldots\times\Gamma_m$ \\
            \hline
            $D$ & Diameter of $\Delta$ \\
		\hline
		$V$ & Calibration parameter \\
            \hline
            $\sigma$ & Regularization parameter, best choice  $\sigma=\sqrt{2}/D$ \\
            \hline
            $\grad \hat{f}_{t+1}(\hat{\boldsymbol{z}}_{t+1})$ & Gradient prediction known at $t$ \\
            \hline
            $\zeta$ & Prediction model average relative error rate \\
            \hline
            $\alpha$ & Minimum ratio the NO must provide for advance resources \\
            \hline
            $\beta$ & Minimum ratio the NO must provide for spot resources \\
		\hline %
		\hline
	\end{tabular}
	\label{table:notation} 
\end{table}

\section{Optimistic Online Learning for Reservation} \label{sec:algo}

\subsection{Algorithm}

Our approach is inspired from the \emph{Follow-the-Regularized-Leader} (FTRL) policy, whereby the learner aims to minimize the loss on all past slots plus a regularization term:
\begin{flalign}
\forall t, \bm z_{t+1} = \arg\min_{\bm z \in \Delta^2} \sum_{i=1}^t f_i(\bm z) + R(\bm z)
\end{flalign}
Due to the convexity of $f_t$, the following property holds:
\begin{flalign}
f_t(\bm z_t) - f_t(\bm z^*) \leq \grad f_t(\bm z_t)^\top (\bm z_t - \bm z^*)
\end{flalign}
which means that the regret against the functions $\{f_t\}$ is upper-bounded by the regret against their linearized form $\bar f_t(\bm z) = \grad f_t(\bm z_t)^\top \bm z$ \cite{mcmahan}. Consequently, the FTRL algorithm simplifies to:
\begin{flalign} \label{eq:ftrl}
\forall t, \bm z_{t+1} = \arg\min_{\bm z \in \Delta^2} \sum_{i=1}^t \grad f_i(\bm z_i)^\top \bm z + R(\bm z)
\end{flalign}

In our approach, we consider an additional gradient term, which is the \emph{optimistic} next slot gradient prediction $\grad \hat f_{t+1}( \hat{\boldsymbol{z}}_{t+1})$. 
In the FTRL, the regularization function is quadratic $R(\bm z) = \frac{1}{2\eta}||\bm z||^2$. In contrast, we design a sequence of proximal regularizers:
\begin{align}
    \forall t=1\ldots T, \quad r_t(\bm z) = \frac{\sigma_t}{2}||\bm z -\bm z_t||^2, \label{eq:reg}
\end{align}
with $||.||$ the Euclidean norm. The regularizer parameters are:
\begin{align}
    \sigma_t &= \sigma\Big( \sqrt{h_{1:t}} - \sqrt{h_{1:t-1}}\Big), \label{eq:acc_error}\\
    h_t &= ||\grad f_t(\bm z_t) - \grad \hat f_t(\hat{\boldsymbol{z}_t})  ||^2, \label{eq:quadratic_error}
\end{align}
where $\sigma\geq 0$, and $h_{1:t}=\sum_{i=1}^t h_i$.

All the above lead to the final form of our algorithm decision step:
\begin{align}
    \bm z_{t+1} = \arg& \min_{\bm z \in \Delta^2} \Big\{ r_{1:t}(\bm z) + \notag\\
    &\Big(\sum_{i=1}^t \grad f_i(\bm z_i) + \grad \hat f_{t+1}(\hat{\boldsymbol{z}}_{t+1})\Big)\top \bm z\Big\} \label{eq:stepOOLR}
\end{align}

\begin{algorithm}[t]
\SetAlgoRefName{OOLR}
\caption{Optimistic Online Learning for Reservation}
\DontPrintSemicolon
\KwInitialize{ \; $\bm z_1 \in \Delta^2, \sigma = 1$, $a_1$, $\bm q_1$, $f_1(\bm z_1)$ } 
\For{ $t=1,\ldots, T-1$ } 
{
Observe the new prediction of the gradient $\grad \hat{f}_{t+1}(\hat{\boldsymbol{z}}_{t+1})$ \;
Decide $\bm z_{t+1}$ by solving \eqref{eq:stepOOLR} \;
Observe the demand $a_{t+1}$, the reservation price $\bm p_{t+1}$, the spot price $\bm q_{t+1}$, the contributions $\bm \theta_{t+1}$ \;
Calculate $f_{t+1}(\bm z_{t+1})$ and $\grad f_{t+1}(\bm z_{t+1})$\;
Update $r_{1:t+1}(\bm z)$ according to \eqref{eq:reg} and \eqref{eq:acc_error} \;
}
\end{algorithm}

\subsection{Performance analysis}
We start with the necessary assumptions.
\begin{assumption}
The sets $\Gamma_i$, $i=1\ldots m$, are convex and compact, and it holds $|x|\leq D_i$, for any $x \in \Gamma_i$\footnote{Note that we can rename the set $\Gamma_1\times\ldots\times\Gamma_m$ as $\Delta$ and simply assume $\Delta$ is a compact convex set with diameter $D$. The design of the different sets $\Gamma_i$ allows us to choose a different reservation restriction for each resource type.}.
\end{assumption}

\begin{assumption}
The function $f_t$ is convex.
\end{assumption}

\begin{assumption}
$\{r_t\}_{t=1}^T$ is a sequence of proximal non-negative functions.
\end{assumption}

\begin{assumption}
Prediction $\grad \hat{f}_{t+1}(\hat{\boldsymbol{z}}_{t+1})$ is known at $t$.
\end{assumption}


\begin{corollary} \label{corollary-oolr}
Under Assumptions 1-4, we derive from \cite[Theorem~1]{mohri} and \cite[Theorem~1]{mhaisen} the following regret bound:
\begin{align}
    \boxed{
    R(T) \leq \sqrt{\sum_{t=1}^T ||\grad f_t(\bm z_t)-\grad \hat{f}_t(\hat{\boldsymbol{z}}_t)||^2}(\frac{2}{\sigma}+\frac{\sigma}{2}2D^2)} \label{eq:regret_bound}
    \end{align}
\end{corollary}

\begin{proof}
First let's remark that the function $h_{0:t}:\bm z \rightarrow r_{0:t}(\bm z) + (c_{1:t}+\tilde c_{t+1})^\top \bm z$ is $1$-strongly convex, with respect to the norm $||.||_{(t)}$. It allows us to use \cite[Theorem~1]{mohri}, which yields regret:
\begin{align}
    R(T) \leq r_{1:T}(\bm z^*) + \sum_{t=1}^T ||c_t-\tilde c_t||^2_{(t),*} \quad \forall \bm z^* \in \Delta^2 \label{lemma}
\end{align}

Now, we define the norm $||x||_{(t)} = \sqrt{\sigma_{1:t}}||x||$, which has dual norm $||x||_{(t),*} = ||x||/\sqrt{\sigma_{1:t}}$.
We remark that $\sigma_{1:t}=\sigma\sqrt{h_{1:t}}$, and starting from \eqref{lemma}, we get:
\begin{align}
    R(T) &\leq \frac{\sigma}{2}\sum_{t=1}^T (\sqrt{h_{1:t}}-\sqrt{h}_{1:t-1})||\bm z^* - \bm z_t||^2
    + \sum_{t=1}^T \frac{h_t}{\sigma\sqrt{h_{1:t}}} \\ \notag
     &\leq \frac{\sigma}{2}\sum_{t=1}^T (\sqrt{h_{1:t}}-\sqrt{h}_{1:t-1})2D^2
    + \sum_{t=1}^T \frac{h_t}{\sigma\sqrt{h_{1:t}}} \notag
\end{align}

We use the first order definition of convexity on the square root function to get:
\begin{align}
    \sqrt{h_{1:t}} - \sqrt{h_{1:t-1}} &\leq    \notag \frac{1}{2\sqrt{h_{1:t}}}(h_{1:t} - h_{1:t-1})\\
    &= \frac{h_t}{2\sqrt{h_{1:t}}} \notag
\end{align}
Thus,
\begin{align}
    R(T) \leq \frac{\sigma}{4}\sum_{t=1}^T \frac{h_t}{\sqrt{h_{1:t}}}2D^2 + \sum_{t=1}^T \frac{h_t}{\sigma\sqrt{h_{1:t}}} \label{intermed}
\end{align}

From \cite[Lemma~3.5]{cesa}, we have:

\begin{align}
    \sum_{t=1}^T \frac{h_t}{\sqrt{h_{1:t}}}  \leq 2\sqrt{h_{1:t}}
\end{align}

Plugging this result into \eqref{intermed}, it yields:

\begin{align}
    R(T) \leq \sqrt{h_{1:t}} (\frac{2}{\sigma} + \frac{\sigma}{2}2D^2)
\end{align}

\end{proof}

\emph{Remark 1.} We observe that a certain value of $\sigma$ can minimize the upper-bound on the regret, but one has to know the diameter of the decision set $\sqrt{2}D$. The very value of $\sigma$ which minimizes the upper-bound is:
\begin{align}
    \sigma = \frac{\sqrt{2}}{D} \label{eq:sig}
\end{align}

We re-write the upper bound:

\begin{align}
    \boxed{
    R(T) \leq 2\sqrt{2}D\sqrt{\sum_{t=1}^T ||\grad f_t(\bm z_t)-\grad \hat{f}_t(\hat{\boldsymbol{z}}_t)||^2}
   } \label{eq:regret_bound2}
\end{align}

\emph{Remark 2.} The regret bound is in $\mathcal{O}(\sqrt{T})$ if predictions are arbitrarily bad i.e. $\sum_{t=1}^T ||\grad f_t(\bm z_t)-\grad \hat{f}_t(\hat{\boldsymbol{z}}_t)||^2 = \mathcal{O}(T)$, and becomes \emph{null} when the predictions are perfect, i.e. when $\forall t, \quad  \grad \hat{f}_t(\hat{\boldsymbol{z}}_t) = \grad f_t(\boldsymbol{z}_t)$.

\emph{Remark 3.} We implement an online learning prediction method that learns how to predict the gradient with the regret $\mathcal{O}(2mGM\sqrt{T})$, where $G$ and $M$ are key constant in \cite{anava}. 
Other prediction methods could be applied to the prediction of the gradient; however, this online learning method offers sublinear regret guarantees against all types of traces, even non-stationary.

\emph{Conclusion.} We conclude that our OOLR algorithm brings the best of both worlds. Given arbitrarily bad predictions, it provides the same guarantee of sublinear regret as the FTRL algorithm, i.e. $\mathcal{O}(\sqrt{T})$. Associated with an accurate prediction model, it provides tighter guarantee of performance down to $\mathcal{O}(1)$ in the ideal case, i.e. when predictions are perfect $\sum_{t=1}^T ||\grad f_t(\bm z_t)-\grad \hat{f}_t(\hat{\boldsymbol{z}}_t)||^2 = \mathcal{O}(1)$.

\subsection{Complexity analysis}
We stress there that the computational cost and memory requirements of the OOLR algorithm are fairly low. We need to solve at each slot $t$, the problem \eqref{eq:stepOOLR}. We add the term $r_t(\bm z)$ to the previous regularizer $r_{1:t-1}(\bm z)$. We can just replace $r_{1:t-1}(\bm z)$ by $r_{1:t}(\bm z)$ in the same variable to limit storage cost. 
The gradient terms $\grad f_i(\bm z_i)$ are equal to:
\begin{align}
    \grad f_i(\bm z_i) &= \begin{bmatrix}
          -V \frac{a_i \theta_{i,1}}{1+ (\bm x_i + \bm y_i)^\top \bm \theta_t  } + p_{t,1} \\
            -V \frac{a_i \theta_{i,2}}{1+ (\bm x_i + \bm y_i)^\top \bm \theta_i  } + p_{i,2} \\
           \vdots \\
           -V \frac{a_i \theta_{i,m}}{1+ (\bm x_i + \bm y_i)^\top \bm \theta_i  } + p_{i,m} \\
           -V \frac{a_i \theta_{i,1}}{1+ (\bm x_i + \bm y_i)^\top \bm \theta_i  } + q_{i,1} \\
           -V \frac{a_i \theta_{i,2}}{1+ (\bm x_i + \bm y_i)^\top \bm \theta_i  } + q_{i,2} \\
           \vdots \\
           -V \frac{a_i \theta_{i,m}}{1+ (\bm x_i + \bm y_i)^\top \bm \theta_i  } + q_{i,m}
         \end{bmatrix}. \label{eq:gradient}
 \end{align}
 Thus, at each slot $i$, we need to store the vectors $\bm p_i$, $\bm q_i$, $\bm \theta_i$ of length $m$ and the scalar $a_i$. Therefore, memory requirements are of $3m+1 = \mathcal{O}(m)$. We can just replace the gradient term $\sum_{i=1}^{t-1} \grad f_i(\bm z_i)$ by $\sum_{i=1}^{t} \grad f_i(\bm z_i)$ in the same variable to limit storage cost. 
The computation of the gradient \eqref{eq:gradient} runs in $2m(m+4)$ operations. The computation of the regularizer term \eqref{eq:reg} runs in $4m$ operations for $||\bm z - \bm z_t||^2$ and $4m+2$ operations for $\sigma_t$. Therefore the running time for \eqref{eq:stepOOLR} is in $2m(m+4)+8m+2=\mathcal{O}(m^2)$.

The complexity of the OOLRgrad algorithm is higher as we must take account of the complexity of the prediction module ARMA-OGD from \cite{anava}. We apply the ARMA-OGD to each gradient item separately. For one item, the prediction consists of an online gradient descent update of the $q$ lag coefficients. Then, the prediction is the linear combination of the $q$ previous real values of the gradient weighted by the $q$ lag coefficients. Thus, the running time of ARMA-OGD applied to our specific case is in $\mathcal{O}(mq)$. The memory requirements of the ARMA-OGD are $2m2q$ as we must store the last $q$ observations of the real gradient and the $q$ lag coefficients, for each item of the gradient. Therefore, memory requirements are of $\mathcal{O}(mq)$.

We conclude that both OOLR and OOLRgrad have fairly low running time and memory requirements, given that the number of resources $m$ composing one slice is not too high, which is practically the case.



\section{Numerical evaluation}

\textbf{Experimental scenario}. We consider a Mobile Virtual Network Operator (MVNO) which aims to acquire network resources that constitute the end-to-end network slice dedicated to its specific network service. Confronted with unknown and evolving traces such as the users' demand, the prices, and contributions of the network resources, the MVNO will follow the online reservation strategy designed by our OOLR solution. We consider the base case where the MVNO faces the incoming demand at one Base Station (BS) and must reserve $m=3$ types of resources to deliver its network service, encompassing radio resources at the BS, backhaul link capacity, and computing resources at the core. This base case falls under the scope of our system model. 

To model user demand, we use a real-world data set that contains the aggregated traffic volumes seen across multiple BSs owned by a major MNO of Shanghai \cite{jb-icc20}. Traffic volumes
have been recorded over a one-month period, spanning from
Friday 1 August 2014 00:00 to Sunday 31 August 2014 23:50,
with each recording averaged over a period of 10 minutes.
Hence, there are 6 measurements per hour and a total of 4464
measurements for each BS over this period.


We assume the network resources prices vary with non-stationary dynamics. We model such variations with an AR($1$) (Auto-Regressive with $1$ lag) process, the discrete-time equivalent of the Ornstein-Ulhenbeck (OU) process. This stochastic process is applied in financial mathematics to model stock prices. 
We model the contribution parameters -items of vector $\bm \theta_t$- as varying and non-stationary. Each item follows a seasonal trend (sine wave), with an offset and added OU stochastic process.

We compare the OOLRgrad solution to the FTRL baseline. The latter consists of the update as defined in equation \eqref{eq:ftrl}. We introduce the parameter $\zeta$ to control the quality of different prediction models, where $\zeta$ is the average relative error rate of the prediction $\grad \hat{f}_{t+1}(\boldsymbol{\hat{z}}_{t+1})$ against the real value $\grad f_{t+1}(\boldsymbol{z}_{t+1})$. We set $\zeta=0, 0.3$ and $4$ to represent prediction models from perfect accuracy to arbitrarily bad. This allows us to introduce three OOLR baselines, with different levels of prediction accuracy.

\textbf{Prediction module}. The solution OOLR is \emph{optimistic} in the sense it allows the SP to use the predicted gradient term $\grad \hat{f}_{t+1}(\boldsymbol{\hat{z}}_{t+1})$ of the next slot. In \eqref{eq:regret_bound}, we concluded that accurate predictions can greatly enhance the performance, as the regret bound goes from $\mathcal{O}(\sqrt{T})$ when predictions are arbitrarily bad to $\mathcal{O}(1)$ when predictions are perfect. This observation paves the way to the introduction of a prediction module, in support of our OOLR decision algorithm. We aim to find an accurate, robust and computationally low model. The algorithm ARMA-OGD created by Anava et al. in \cite{anava} presents these three key advantages. It consists of learning the AR($q$) signal of the trace where the $q$ lag coefficients are updated online at each slot by the gradient descent method. The algorithm guarantees that the total loss is no more on average than the loss of the best ARMA predictor with full hindsight. 

First, we show in Fig. \ref{fig:ARMAOGD} that the model is accurate against two intricate signals. The SP demand is based on multiple latent factors, which makes the signal non-stationary and hard to predict. Yet, we observe the predicted signal is able to track the SP demand. The $2m$ gradient items are composed of multiple signals, namely the SP demand, the prices and contributions of the network resources. 
Yet again, the model is able to give an accurate predicted signal. Secondly, ARMA-OGD provides guarantees of performance against all types of traces, which ensures its robustness. The total squared loss of the model is a $\mathcal{O}(\sqrt{T})+ Res$, where $Res$ represents the residual squared loss of the best ARMA predictor with full hindsight of the target signal. We show in Fig. \ref{fig:squaredloss} the convergence of the average squared loss towards $Res$. Finally, the ARMA-OGD is based on the OGD update step, which is very low computationally and allows us to develop the algorithm alongside the OOLR solution. We insist here that the two combined solutions having both low time complexity allow the SP to take \emph{optimistic decisions in real time}. There exists other models which employ advanced techniques such as Neural Networks that would obtain better accuracy than the ARMA-OGD. Nevertheless, these models necessitate an offline training phase, do not provide guarantees of performance, and have higher time complexity.

\begin{figure}
    \centering
    \includegraphics[width=1.\linewidth]{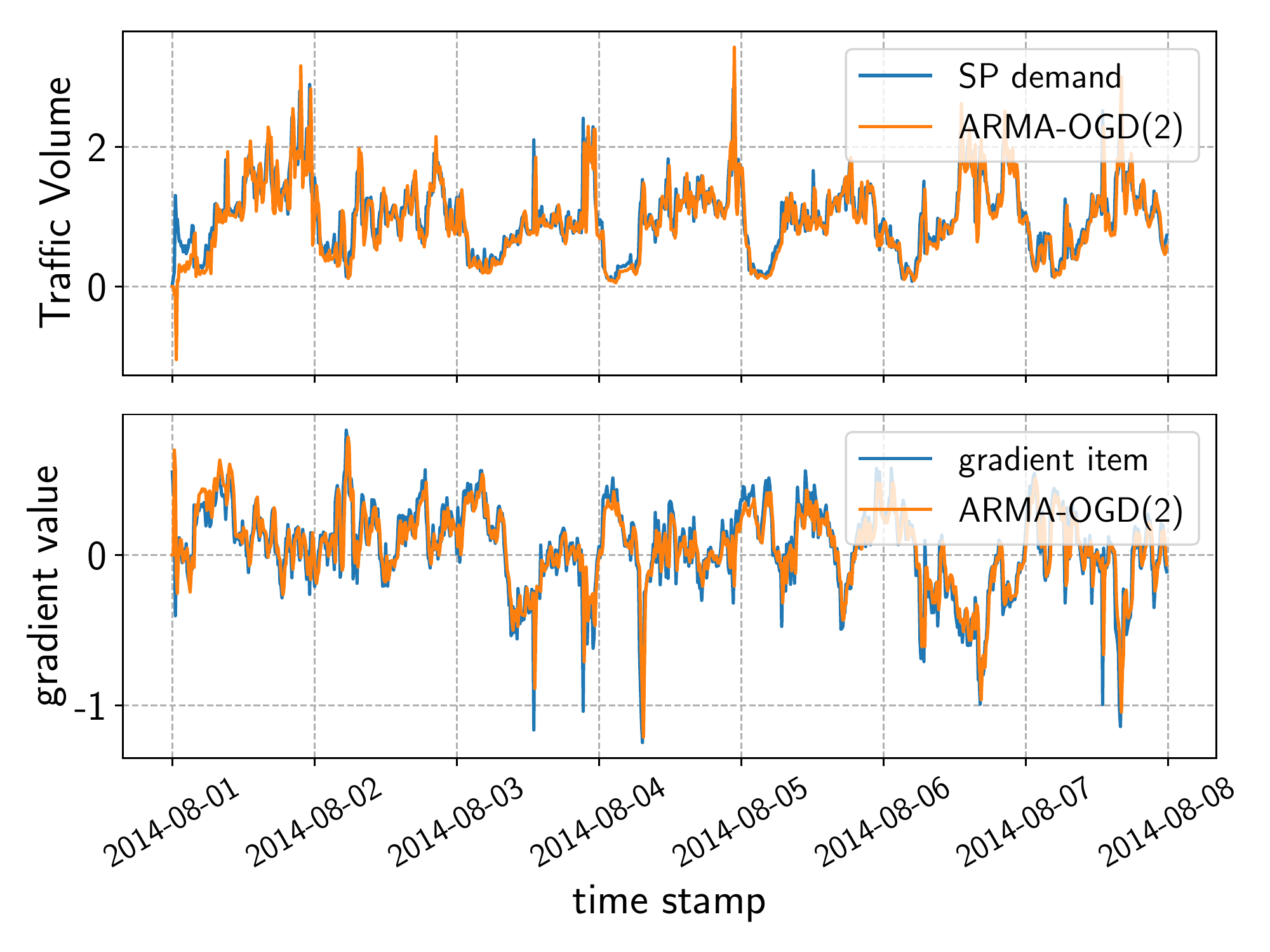}
    \caption{\small{The x-axis encompasses the first week of August period. \emph{Upper part:} the predicted signal against the real-world MVNO demand signal. The y-axis values are normalized. \emph{Lower part:} the predicted signal against the first of the gradient $2m$ items.}}
    \label{fig:ARMAOGD}
    \vspace{-6mm}
\end{figure}

\begin{figure}
    \centering
    \includegraphics[width=1.\linewidth]{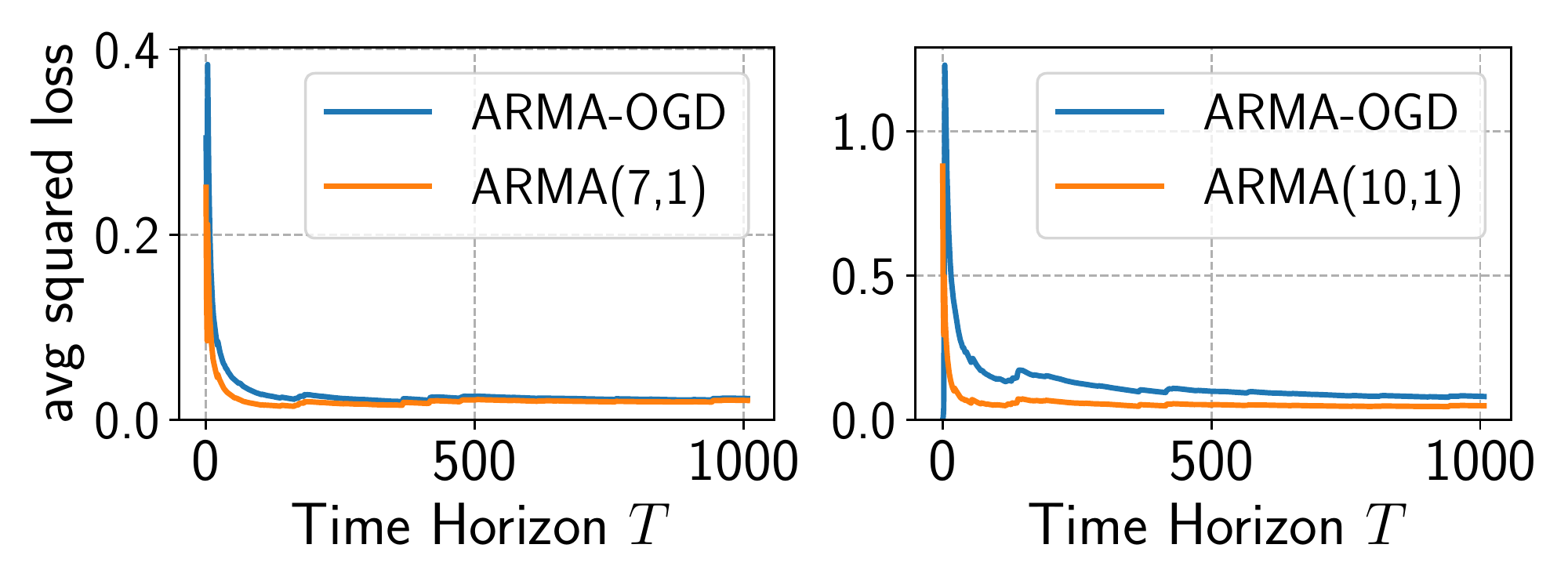}
    \caption{\small{We evaluate the accuracy of the models on the first week of August. \emph{Left side:} We observe the convergence of the average squared loss of the predicted gradient first item towards the best ARMA in hindsight. \emph{Right side:} We observe the convergence of the average squared loss of the predicted MVNO demand toward the best ARMA in hindsight.}}
    \label{fig:squaredloss}
    \vspace{-6mm}
\end{figure}

\textbf{Impact of the quality of predictions}. The SP can reserve $m=3$ kinds of resources. We assume the NO sets the upper bound constraint to $D_i=1,\forall i$. This means the SP reserves normalized values for each type of resource. 
 Our goal is to maximize the SP utility while avoiding excessive reservation cost. We balance between the two terms (utility and cost) using the hyper-parameter $V$. 
 We calibrate $V=2$ to have both terms of the same order.

We call OOLRgrad the online decision algorithm OOLR because the prediction method ARMA-OGD is directly applied to the gradient items. 
We show in Fig. \ref{fig:static} against the static benchmark as defined in \eqref{eq:static_regret} the performance of the OOLRgrad solution, the classical FTRL algorithm with euclidean regularizer and the different OOLR models $\zeta=0, 0.3$ and $4$. We first observe the convergence of the average regret $R_T/T$ towards $0$ for the five models, which confirm the regret bound of $\mathcal{O}(\sqrt{T})$ even for arbitrarily bad predictions (represented by the OOLR $\zeta=4$ model). Secondly, we observe a negative regret for the other four models, which confirm the $\mathcal{O}(1)$ regret bound when the predictions are accurate and the accumulated error $\sum_{t=1}^T ||\grad f_t(\boldsymbol{z}_t) - \grad \hat{f}_t(\hat{\boldsymbol{z}}_t)||^2$ is close to $0$. Zooming in the last slots, we remark that our OOLRgrad solution based on the ARMA-OGD predictor shows better performance than the OOLR solution with a $70\%$ accurate predictor ($\zeta=0.3$) and is inferior to the OOLR with perfect predictor ($\zeta=0$). The OOLRgrad and the OOLR $\zeta=0, 0.3$ solutions outperform the FTRL baseline, which shows that the incorporation of accurate predictions enhances the performance. One needs to be cautious as arbitrarily bad predictions (OOLR $\zeta=4$) worsens the performance.  
In Fig. \ref{fig:dynamic}, we show the performance of the same solutions against the optimal benchmark, defined by the dynamic sequence $\{\bm z_t^* \}$, where $\forall t$, $$\bm z_t^* = \arg\min_{\bm z \in \mathcal{Z}} f_t(\bm z).$$ Against such competitive benchmark, the regret cannot be sublinear and thus the convergence of $R_T/T$ towards $0$ is not achieved. Nevertheless, we observe that the OOLRgrad solution displays good performance when compared to the different baselines.

\begin{figure}
\begin{subfigure}{.24\textwidth}
    \centering
    \includegraphics[width=1.\linewidth]{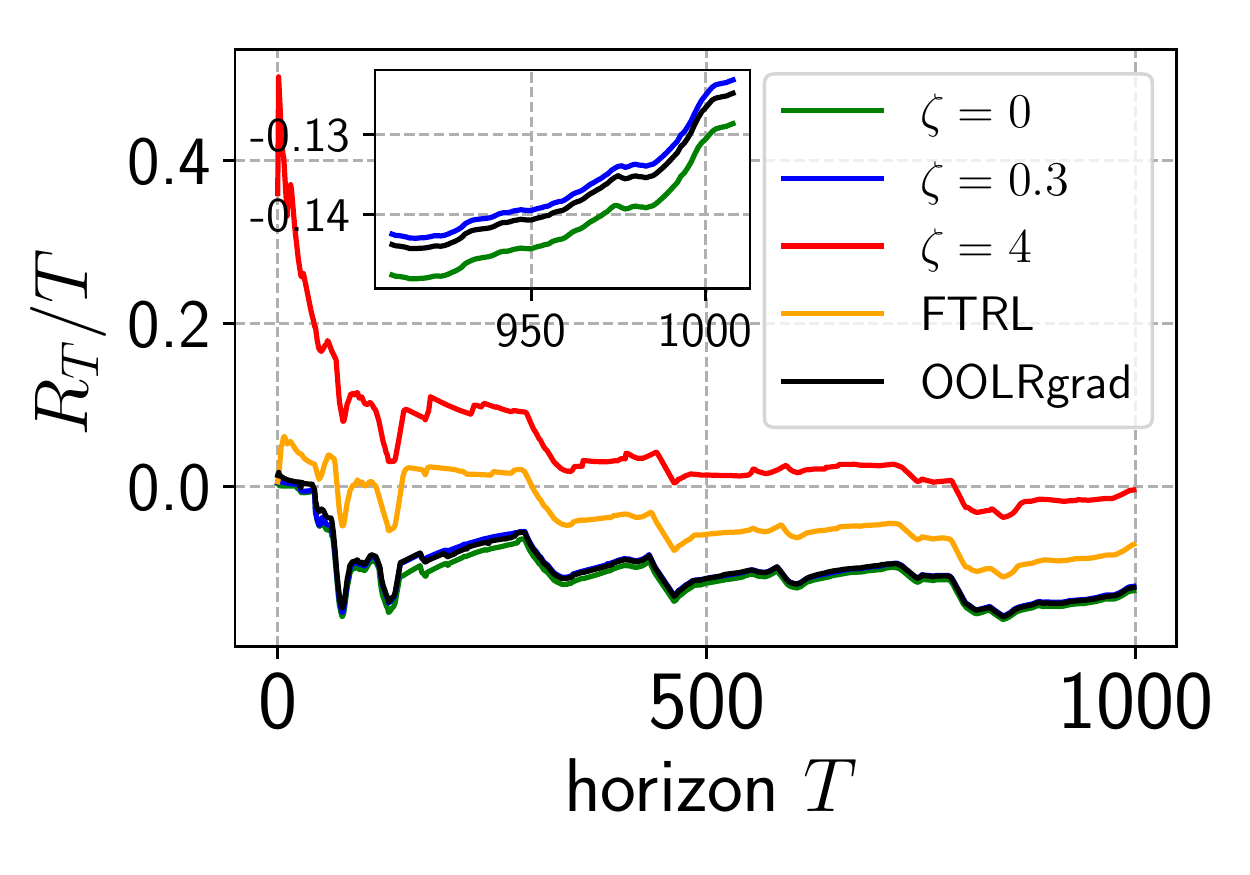}
    \vspace{-6mm}    
    \caption{Static benchmark}
    \label{fig:static}
\end{subfigure}
\begin{subfigure}{.24\textwidth}
    \centering
    \includegraphics[width=1.\linewidth]{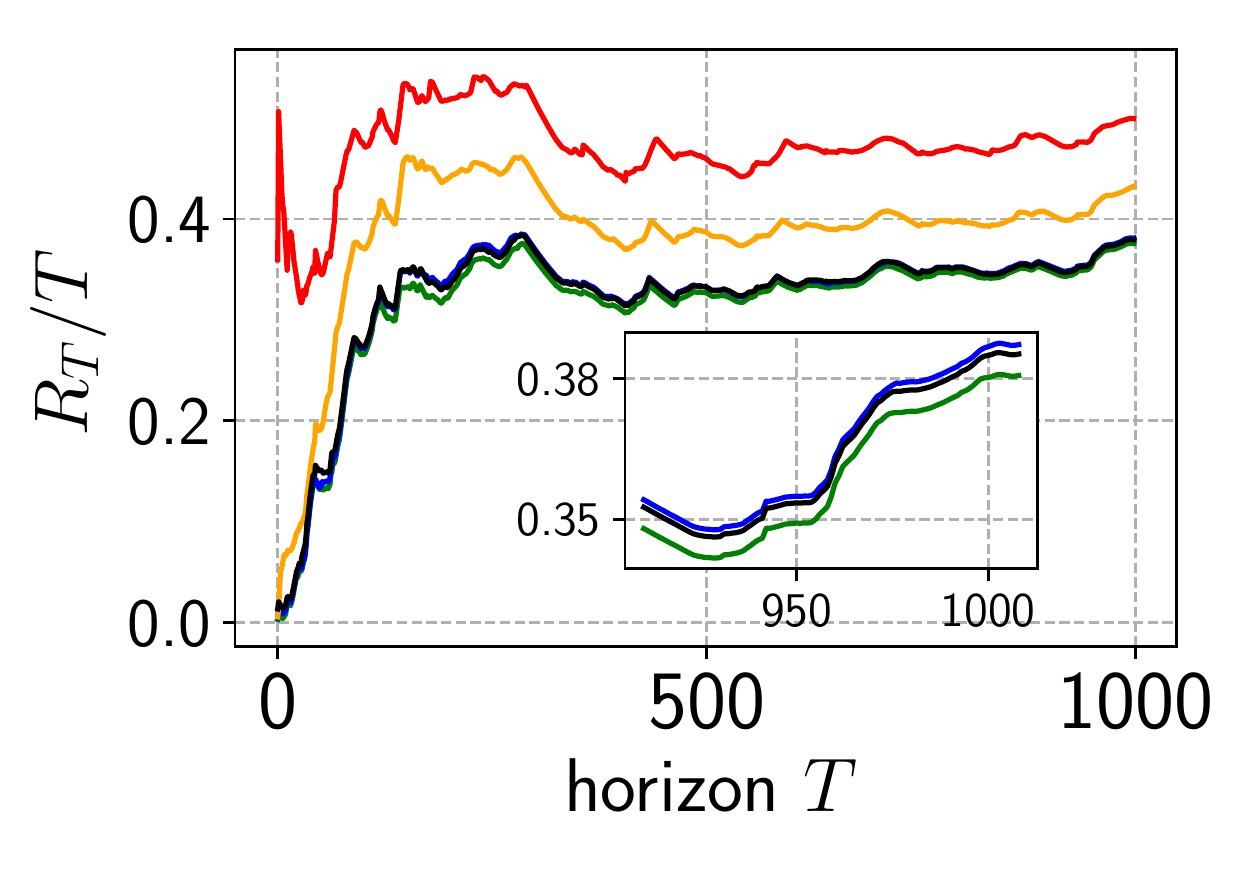}
    \vspace{-6mm}
    \caption{Optimal benchmark}
    \label{fig:dynamic}
\end{subfigure}
\vspace{-3mm}
\caption{\small{\emph{Evolution of $R_T/T$:} Horizon $T=1008$, $m=3$, $V=2$, $\bm D=[1,1,1]$, $D=\sqrt{3}$, $\sigma = \sqrt{2}/D$.}}
\label{fig:regret}
\end{figure}

\textbf{Extension}.
Now we evaluate the OOLRgrad solution in the scenario where the NO is unable to fulfill the SP request in its entirety.
We focus on a basic scenario in which the NO ensures a minimum ratio of $\alpha$ for in-advance resources -- in a more complex scenario the NO commits to a ratio of $\alpha_i$ for each resource $i$, where $\alpha_i$ are possibly different. Thus, for each resource $i$ at slot $t$, the SP expect to receive a ratio $\alpha_{i,t}$ that belongs to the set $[\alpha,1]$. We draw the $\{\alpha_{i,t}\}_t$ from the uniform distribution on $[\alpha,1]$.  We assume that the NO consistently deliver all requested spot resources, thus we keep $\beta=1$.
We observe in Fig. \ref{fig:SLA} the regret performance of the OOLRgrad solution for three different SLAs, which are $\alpha\in\{0.5, 0.8, 0.95\}$. We observe that the performance stays similar regardless of the SLA the SP has complied for, which implies our OOLRgrad solution is consistently applicable.

\begin{figure}
    \centering
    \includegraphics[width=7cm, height=5cm]{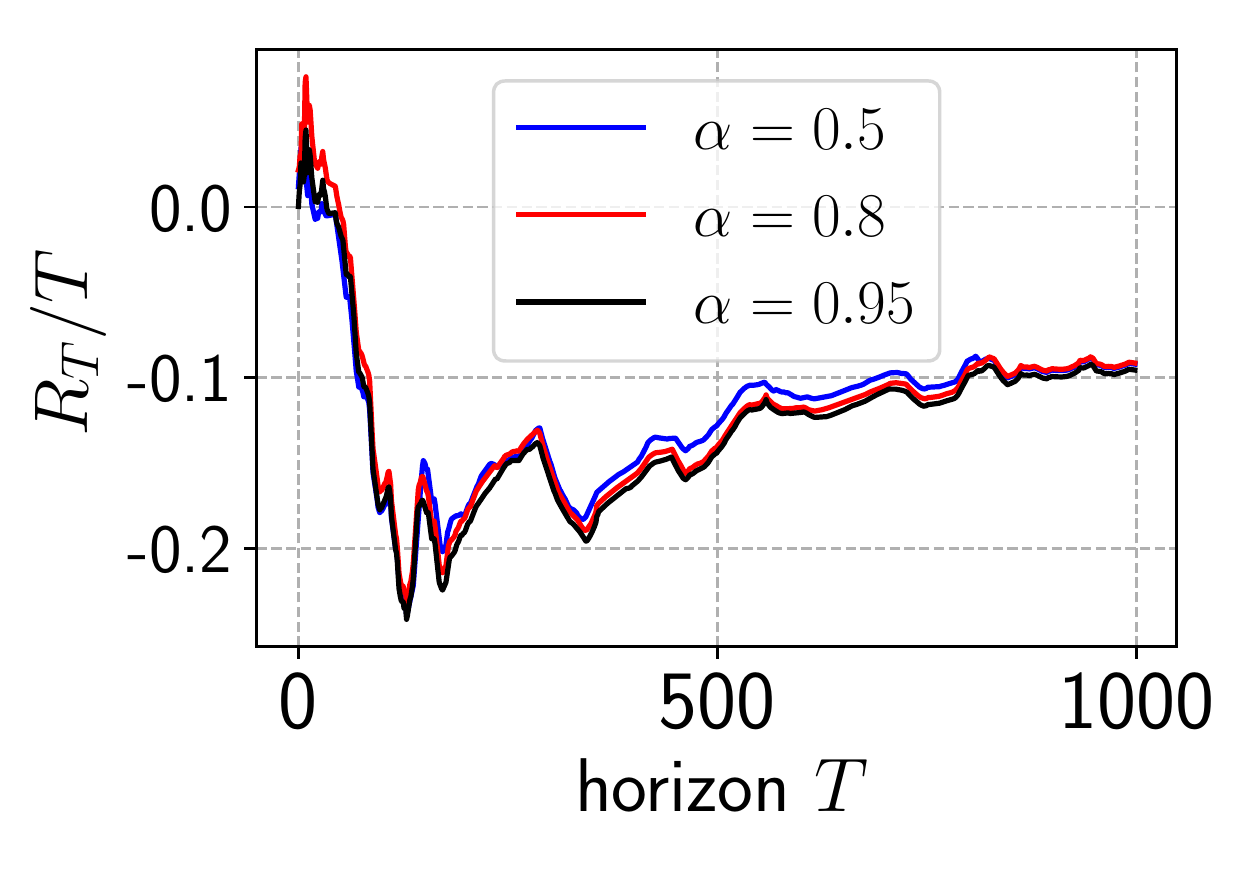}
    \caption{\small{\emph{Evolution of $R_T/T$:} Horizon $T=1008$, $m=3$, $V=2$, $\bm D=[1,1,1]$, $D=\sqrt{3}$, $\sigma = \sqrt{2}/D$.}}
    \label{fig:SLA}
    \vspace{-5mm}
\end{figure}


\vspace{-2mm}
\section{Conclusion} \label{sec:conclusion}

In this paper, we introduced the Optimistic Online Learning for Reservation (OOLR) algorithm that allows the SP to make reservations under uncertainty while incorporating predictions about the future gradient. We then proposed to combine this decision model with a prediction model, thus creating the OOLRgrad solution with better performance than the classical FTRL solution.



\section{Acknowledgments}

The research leading to this work is funded, in part, by Science Foundation Ireland (SFI), the National Natural Science Foundation of China (NSFC), and the European Commission under the SFI-NSFC Partnership Programme Grant Number 17/NSFC/5224, SFI grant 13/RC/2077 P2, and the Grant Number 101017109 (DAEMON).



\vspace{-2mm}

\bibliographystyle{IEEEtran}

\bibliography{ref.bib}

\begin{thebibliography}{10}
\providecommand{\url}[1]{#1}
\csname url@samestyle\endcsname
\providecommand{\newblock}{\relax}
\providecommand{\bibinfo}[2]{#2}
\providecommand{\BIBentrySTDinterwordspacing}{\spaceskip=0pt\relax}
\providecommand{\BIBentryALTinterwordstretchfactor}{4}
\providecommand{\BIBentryALTinterwordspacing}{\spaceskip=\fontdimen2\font plus
\BIBentryALTinterwordstretchfactor\fontdimen3\font minus
  \fontdimen4\font\relax}
\providecommand{\BIBforeignlanguage}[2]{{%
\expandafter\ifx\csname l@#1\endcsname\relax
\typeout{** WARNING: IEEEtran.bst: No hyphenation pattern has been}%
\typeout{** loaded for the language `#1'. Using the pattern for}%
\typeout{** the default language instead.}%
\else
\language=\csname l@#1\endcsname
\fi
#2}}
\providecommand{\BIBdecl}{\relax}
\BIBdecl

\bibitem{5g-ppp}
G.~P. A.~W. Group \emph{et~al.}, ``View on 5g architecture,'' \emph{White
  Paper, July}, 2016.

\bibitem{mano-nfv}
F.~Z. Yousaf \emph{et~al.}, ``Manoaas: A multi-tenant nfv mano for 5g network
  slices,'' \emph{IEEE Communications Magazine}, vol.~57, no.~5, pp. 103--109,
  2019.

\bibitem{foukas-commag}
X.~{Foukas} \emph{et~al.}, ``{Network Slicing in 5G: Survey and Challenges},''
  \emph{IEEE CommMag}, vol.~55, no.~5, pp. 94--100, 2017.

\bibitem{amazon-reserved}
\BIBentryALTinterwordspacing
``{Amazon EC2},'' 2021, {Reserved Instances}. [Online]. Available:
  \url{https://aws.amazon.com/ec2/purchasing-options/reserved-instances/}
\BIBentrySTDinterwordspacing

\bibitem{amazon-spot}
\BIBentryALTinterwordspacing
``{Amazon EC2},'' 2021, {Spot Instances}. [Online]. Available:
  \url{https://aws.amazon.com/ec2/spot/}
\BIBentrySTDinterwordspacing

\bibitem{carlee-infocom18}
M.~{Khodak} \emph{et~al.}, ``{Learning Cloud Dynamics to Optimize Spot Instance
  Bidding Strategies},'' in \emph{IEEE INFOCOM 2018 - IEEE Conference on
  Computer Comm.}, 2018, pp. 2762--2770.

\bibitem{carlee-sigcomm15}
{L. Zheng, C. Joe-Wong, C. W. Tang, M. Chiang, and X. Wang}, ``{How to Bid the
  Cloud},'' in \emph{Proc. of ACM SIGCOMM}, 2015.

\bibitem{DeepCog}
D.~Bega \emph{et~al.}, ``Deepcog: Cognitive network management in sliced 5g
  networks with deep learning,'' in \emph{IEEE INFOCOM 2019 - IEEE Conference
  on Computer Comm.}, 2019, pp. 280--288.

\bibitem{infocom17}
J.~Wang \emph{et~al.}, ``Spatiotemporal modeling and prediction in cellular
  networks: A big data enabled deep learning approach,'' in \emph{IEEE INFOCOM
  2017-IEEE Conference on Computer Comm.}\hskip 1em plus 0.5em minus
  0.4em\relax IEEE, 2017, pp. 1--9.

\bibitem{mobihoc18}
C.~Zhang and P.~Patras, ``Long-term mobile traffic forecasting using deep
  spatio-temporal neural networks,'' in \emph{Proc. of the Eighteenth ACM
  International Symposium on Mobile Ad Hoc Networking and Computing}, 2018, pp.
  231--240.

\bibitem{oliveiraTNSM}
D.~H. Oliveira, T.~P. de~Araujo, and R.~L. Gomes, ``An adaptive forecasting
  model for slice allocation in softwarized networks,'' \emph{IEEE TNSM},
  vol.~18, no.~1, pp. 94--103, 2021.

\bibitem{XaviINFOCOM17}
V.~Sciancalepore \emph{et~al.}, ``Mobile traffic forecasting for maximizing 5g
  network slicing resource utilization,'' in \emph{IEEE INFOCOM 2017-IEEE
  Conference on Computer Comm.}\hskip 1em plus 0.5em minus 0.4em\relax IEEE,
  2017, pp. 1--9.

\bibitem{XaviTrans19}
V.~Sciancalepore, X.~Costa-Perez, and A.~Banchs, ``Rl-nsb: Reinforcement
  learning-based 5g network slice broker,'' \emph{IEEE/ACM Transactions on
  Networking}, vol.~27, no.~4, pp. 1543--1557, 2019.

\bibitem{cui-iccc20}
Y.~Cui \emph{et~al.}, ``Machine learning based resource allocation strategy for
  network slicing in vehicular networks,'' in \emph{International Conference on
  Comm. in China (ICCC)}.\hskip 1em plus 0.5em minus 0.4em\relax IEEE, 2020,
  pp. 454--459.

\bibitem{jb-icc20}
J.~{Monteil} \emph{et~al.}, ``{Resource Reservation within Sliced 5G Networks:
  A Cost-Reduction Strategy for Service Providers},'' in \emph{IEEE ICC
  Workshops}, 2020, pp. 1--6.

\bibitem{reyhanian}
N.~Reyhanian, H.~Farmanbar, and Z.-Q. Luo, ``Data-driven adaptive network
  resource slicing for multi-tenant networks,'' in \emph{International
  Conference on Acoustics, Speech and Signal Processing}.\hskip 1em plus 0.5em
  minus 0.4em\relax IEEE, 2021, pp. 4715--4719.

\bibitem{zhang-tcom2018}
{Y. Zhang, S. Bi, and Y. J. Angela Zhang}, ``{Joint Spectrum Reservation and
  On-demand Request for Mobile Virtual Network Operators},'' \emph{{IEEE Trans.
  on Comm.}}, vol.~66, 2018.

\bibitem{vincent-TVT2020}
H.~{Zhang} and V.~W.~S. {Wong}, ``{A Two-Timescale Approach for Network Slicing
  in C-RAN},'' \emph{IEEE Trans. on Vehicular Technology}, vol.~69, no.~6, pp.
  6656--6669, 2020.

\bibitem{JB2021}
J.-B. Monteil, G.~Iosifidis, and L.~DaSilva, ``No-regret slice reservation
  algorithms,'' in \emph{ICC 2021}.\hskip 1em plus 0.5em minus 0.4em\relax
  IEEE, 2021, pp. 1--7.

\bibitem{JB2022}
J.-B. Monteil, G.~Iosifidis, and L.~Da~Silva, ``Learning-based reservation of
  virtualized network resources,'' \emph{IEEE TNSM}, 2022.

\bibitem{zinkevich}
{M. Zinkevich}, ``{Online Convex Programming and Generalized Infinitesimal
  Gradient Ascent},'' in \emph{{Proc. of ICML}}, 2003.

\bibitem{ftrl}
S.~Shalev-Shwartz and Y.~Singer, ``A primal-dual perspective of online learning
  algorithms,'' \emph{Machine Learning}, vol.~69, no.~2, pp. 115--142, 2007.

\bibitem{sridharan}
A.~Rakhlin and K.~Sridharan, ``Online learning with predictable sequences,'' in
  \emph{Conference on Learning Theory}.\hskip 1em plus 0.5em minus 0.4em\relax
  PMLR, 2013, pp. 993--1019.

\bibitem{anava}
O.~Anava \emph{et~al.}, ``Online learning for time series prediction,'' in
  \emph{Conference on learning theory}.\hskip 1em plus 0.5em minus 0.4em\relax
  PMLR, 2013, pp. 172--184.

\bibitem{paschos-infocom18}
M.~Leconte \emph{et~al.}, ``{A Resource Allocation Framework for Network
  Slicing},'' in \emph{Proc. of IEEE INFOCOM}, 2018, pp. 2177--2185.

\bibitem{srikant}
S.~G. Shakkottai and R.~Srikant, \emph{Network optimization and control}.\hskip
  1em plus 0.5em minus 0.4em\relax Now Publishers Inc, 2008.

\bibitem{mcmahan}
H.~B. McMahan, ``A survey of algorithms and analysis for adaptive online
  learning,'' \emph{The Journal of Machine Learning Research}, vol.~18, no.~1,
  pp. 3117--3166, 2017.

\bibitem{mohri}
M.~Mohri and S.~Yang, ``Accelerating online convex optimization via adaptive
  prediction,'' in \emph{Artificial Intelligence and Statistics}.\hskip 1em
  plus 0.5em minus 0.4em\relax PMLR, 2016, pp. 848--856.

\bibitem{mhaisen}
N.~Mhaisen, G.~Iosifidis, and D.~Leith, ``Online caching with optimistic
  learning,'' \emph{arXiv preprint arXiv:2202.10590}, 2022.

\bibitem{cesa}
P.~Auer, N.~Cesa-Bianchi, and C.~Gentile, ``Adaptive and self-confident on-line
  learning algorithms,'' \emph{Journal of Computer and System Sciences},
  vol.~64, no.~1, pp. 48--75, 2002.

\end{thebibliography}


\end{document}